\newif\ifprocs
\procstrue
\procsfalse

\newif\ifarxiv
\arxivtrue

\ifprocs 
\documentclass[a4paper,USenglish]{lipics-v2019}
\else
\documentclass[11pt,USenglish]{article}
\usepackage[margin=1.0in]{geometry}
\usepackage[T1]{fontenc}
\usepackage[utf8]{inputenc}
\usepackage[english]{babel}
\usepackage{authblk}
\fi

\usepackage{amsmath}
\usepackage{amsthm}
\usepackage{amssymb}
\usepackage{complexity}
\usepackage{graphicx}
\usepackage{xspace}
\usepackage[export]{adjustbox}
\ifprocs
\nolinenumbers 
\else
\usepackage[colorlinks,linkcolor=blue,citecolor=blue,filecolor=blue,urlcolor=blue]{hyperref}
\fi

\usepackage{color}
\usepackage{amsmath}
\usepackage{mwe}
\usepackage{xcolor}
\usepackage{complexity}
\usepackage{graphicx}
\usepackage[export]{adjustbox}
\usepackage{hyperref}
\usepackage{algorithm}
\usepackage[noend]{algpseudocode}

\algdef{SE}[SUBALG]{Indent}{EndIndent}{}{\algorithmicend\ }%
\algtext*{Indent}
\algtext*{EndIndent}

\ifprocs
\theoremstyle{plain}
\else
\fi

\newtheorem{thm}{Theorem}
\numberwithin{thm}{section}
\newtheorem{lma}[thm]{Lemma}

\newtheorem{assumption}[thm]{Assumption}

\newtheorem{conjecture}[thm]{Conjecture}
\newtheorem{open}[thm]{Open Problem}
\newtheorem{observation}[thm]{Observation}

\makeatletter
\newtheorem*{rep@theorem}{\rep@title}
\newcommand{\newreptheorem}[2]{%
\newenvironment{rep#1}[1]{%
 \def\rep@title{#2 \ref{##1}}%
 \begin{rep@theorem}}%
 {\end{rep@theorem}}}
\makeatother

\newreptheorem{theorem}{Theorem}
\newreptheorem{lemma}{Lemma}

\newcommand{\ProblemName}[1]{\textsf{#1}}

\newcommand{\SCO}{\ProblemName{Set Cover}\xspace}
\newcommand{\PPC}{\ProblemName{p-Partial Cover}\xspace}
\newcommand{\SI}{\ProblemName{Subgraph Isomorphism}\xspace}

\newcommand{\EXC}{\ProblemName{Exact Cover}\xspace}

\newcommand{\dSCO}{\ProblemName{$\Delta$-Set Cover}\xspace}

\newcommand{\KTR}{\ProblemName{kTree}\xspace}
\newcommand{\DKTR}{\ProblemName{Directed kTree}\xspace}
\newcommand{\NTR}{\ProblemName{nTree}\xspace}
\newcommand{\DNTR}{\ProblemName{Directed nTree}\xspace}
\newcommand{\DHAM}{\ProblemName{Directed Hamiltonicity}\xspace}

\renewcommand{\algorithmiccomment}[1]{\bgroup\hfill//~#1\egroup}

\providecommand{\card}[1]{\lvert#1\rvert}
\newcommand{\Ssml}{S_{\textrm{sml}}}
\newcommand{\Sbig}{S_{\textrm{big}}}

\title{The Set Cover Conjecture and Subgraph Isomorphism with a Tree Pattern%
\footnote{This paper is based on two preliminary versions \href{http://arxiv.org/abs/1711.08041}{arXiv:1711.08041} and \href{http://arxiv.org/abs/1708.07591}{arXiv:1708.07591}. }
}
\ifprocs
\titlerunning{The Set Cover Conjecture and Subgraph Isomorphism with a tree pattern}

\EventEditors{Rolf Niedermeier and Christophe Paul}
\EventNoEds{2}
\EventLongTitle{36th International Symposium on Theoretical Aspects of Computer Science (STACS 2019)}
\EventShortTitle{STACS 2019}
\EventAcronym{STACS}
\EventYear{2019}
\EventDate{March 13--16, 2019}
\EventLocation{Berlin, Germany}
\EventLogo{}
\SeriesVolume{126}
\ArticleNo{42}
\else
\fi

\ifprocs
\author{Robert Krauthgamer}
{Weizmann Institute of Science, Rehovot, Israel}
{robert krauthgamer@weizmann.ac.il}{}
{Work supported in part by 
  the Israel Science Foundation grant \#1086/18,
  ONR Award N00014-18-1-2364, 
  a Minerva Foundation grant,
  and a Google Faculty Research Award.
  Part of this work was done while was visiting the Simons Institute for the Theory of Computing.
}

\author{Ohad Trabelsi}{Weizmann Institute of Science, Rehovot, Israel}{ohad.trabelsi@weizmann.ac.il}{}{}
\else

\author[1]{Robert Krauthgamer%
  \thanks{Work supported in part by 
  the Israel Science Foundation grant \#1086/18,
  ONR Award N00014-18-1-2364, 
  a Minerva Foundation grant,
  and a Google Faculty Research Award.
  Part of this work was done while was visiting the Simons Institute for the Theory of Computing.
  }}
\author[2]{Ohad Trabelsi \thanks{Work partly done at IBM Almaden.}}
\affil[1]{Weizmann Institute of Science. Email: 
		 \texttt{robert.krauthgamer@weizmann.ac.il}}
 \affil[2]{Weizmann Institute of Science. Email: \texttt{ohad.trabelsi@weizmann.ac.il}}

\fi

\ifprocs
\authorrunning{R. Krauthgamer and O. Trabelsi}   

\Copyright{Robert Krauthgamer and Ohad Trabelsi}

\else
\fi

\begin{document}
\maketitle

\begin{abstract}
In the \SCO problem, the input is a ground set of $n$ elements and a collection of $m$ sets, and the goal is to find the smallest sub-collection of sets whose union is the entire ground set. 
The fastest algorithm known runs in time $O(mn2^n)$ [Fomin et al., WG 2004], and the Set Cover Conjecture (SeCoCo) [Cygan et al., TALG 2016] 
asserts that for every fixed $\varepsilon>0$, 
no algorithm can solve \SCO in time $2^{(1-\varepsilon)n}\poly(m)$,
even if set sizes are bounded by $\Delta=\Delta(\varepsilon)$.
We show strong connections between this problem and \KTR, a special case of \SI where 
the input is an $n$-node graph $G$ and a $k$-node tree $T$, 
and the goal is to determine whether $G$ has a subgraph isomorphic to $T$.

First, we propose a weaker conjecture Log-SeCoCo, 
that allows input sets of size $\Delta=O(1/\varepsilon \cdot\log n)$, 
and show that an algorithm breaking Log-SeCoCo would imply a faster algorithm 
than the currently known $2^n\poly(n)$-time algorithm [Koutis and Williams,  TALG 2016] for \DNTR, which is \KTR with $k=n$ and arbitrary directions to the edges of $G$ and $T$. This would also improve the running time for \DHAM, for which no algorithm significantly faster than $2^n\poly(n)$ is known despite extensive research.

Second, we prove that 
if \PPC, a parameterized version of \SCO that requires covering at least $p$ elements, cannot be solved significantly faster than $2^{n}\poly(m)$
(an assumption even weaker than Log-SeCoCo) 
then \KTR cannot be computed significantly faster than $2^{k}\poly(n)$, 
the running time of the Koutis and Williams' algorithm.

\end{abstract}
 

\section{Introduction}
\SCO and \SI are two of the most well-researched problems in theoretical computer science. In this paper we show a strong connection between their time complexity. We first discuss each, and then show our results.
\paragraph*{\SCO.} In the \SCO problem, the input is a ground set $[n]=\{1,...,n\}$ and a collection of $m$ sets, and the goal is to find the smallest sub-collection of sets whose union is the entire ground set. 
An exhaustive search takes $O(n2^m)$ time, 
and a dynamic-programming algorithm has running time $O(mn2^n)$~\cite{fomin04}, 
which is faster when $m>n$, a common assumption that we will make throughout. 
In spite of extensive effort, no algorithm that runs in time $O^*(2^{(1-\varepsilon)n})$ is known, although some improvements are known in special cases~\cite{Koivisto09, Bjorklund09, neder16, bjor17}. 
Here and throughout, $O^*(\cdot)$ hides polynomial factors in the instance size, 
and unless stated otherwise, $\varepsilon>0$ denotes a fixed constant
(and similarly $\varepsilon'$). 
Thus, it was conjectured that the above running time is optimal~\cite{cygan16},
even if the input sets are small.
To state this more formally, 
let \dSCO denote the \SCO problem where all sets have size at most $\Delta>0$. 
\begin{conjecture}[Set Cover Conjecture (SeCoCo)~\cite{cygan16}]
For every fixed $\varepsilon>0$ there is $\Delta(\varepsilon)>0$, such that 
no algorithm (even randomized) solves \dSCO in time $O^*(2^{(1-\varepsilon)n})$. 
\end{conjecture}
This conjecture clearly implies that for every $\Delta=\omega(1)$, no algorithm solves \dSCO in time $O^*(2^{(1-\varepsilon)n})$. 
Several conditional lower bounds were based on this conjecture (by reducing \SCO to it) in the recent decade, including for \ProblemName{Steiner Tree}, \ProblemName{Set Partitioning}, and more ~\cite{cygan16, bjor16, bjor15, kow16, kri17}.
%
%
%
The authors of~\cite{cygan16} asked whether the problems they reduce \SCO to can be reduced \textit{back} to \SCO, so that their running time complexity would stand and fall with SeCoCo.
They believed it would be hard to do, since it would probably provide for those problems an alternative algorithm with running time that matches the currently fastest one, which is very complex and took decades to achieve for some (e.g., for \ProblemName{Steiner Tree}). 
\paragraph*{Connection to SETH.}
No formal connection is known to date between the SeCoCo conjecture 
and the Strong Exponential Time Hypothesis (SETH) of~\cite{ImpaSETH}, 
which asserts that for every $\varepsilon>0$ there exists $k(\varepsilon)$, such that \ProblemName{kSAT} on $N$ variables and $M$ clauses cannot be solved in time $O^*(2^{(1-\varepsilon)N})$.
Cygan et al.~\cite{cygan16}
provided a partial answer by showing a SETH-based  lower bound 
for a certain variant of \SCO (that counts the number of solutions).
It is known that the weaker assumption ETH implies a $2^{\Omega(n)}$ time 
lower bound for \SCO, even if $\Delta=O(1)$,
and that \ProblemName{SAT} can be solved in time $O^*(2^{(1-\varepsilon)N})$ 
if and only if \SCO can be solved in time $O^*(2^{(1-\varepsilon')m})$,
see~\cite{cygan16}. 
Some researchers hesitate to rely on SeCoCo as a conjecture, 
and prefer other, more popular conjectures such as SETH. 
For example, a running time lower bound for \ProblemName{Subset Sum}
was recently shown~\cite{abboud2017seth} based on SETH,
even though a lower bound based on SeCoCo was already known~\cite{cygan16}.

We address the necessity of SeCoCo by proposing a weaker assumption, and showing an independent justification for it. Our conjecture deals with \dSCO for $\Delta=O(\log n)$, as follows.

\begin{conjecture}[Logarithmic Set Cover Conjecture (Log-SeCoCo)]
For every fixed $\varepsilon>0$, there is $\Delta(\varepsilon,n)=O(1/\varepsilon\cdot\log n)$ such that no algorithm (even randomized) solves \dSCO in time $O^*(2^{(1-\varepsilon)n})$.
\end{conjecture}

The fastest algorithm known for \dSCO runs in time $O^*(2^{n\lambda_{\Delta}})$~\cite{Koivisto09} for $\lambda_{\Delta}=(2\Delta-2)/\sqrt{(2\Delta-1)^2-2\ln(2)}\leq 1-1/(2\Delta)$, where the inequality assumes $\Delta\geq 2$, 
hence this running time is slightly faster than for general \SCO. 
All known hardness results that are based on SeCoCo can be based also on our conjecture, with appropriate adjustments related to the set sizes in \SCO parameterized by the universe size plus the solution size~\cite{cygan16} and in Parity of Set Covers~\cite{bjor15}.

\paragraph*{\SI with a tree pattern.} 
The \SI problem asks whether a host graph $G$ contains a copy of a pattern graph $H$ as a subgraph. It is well known to be NP-hard since it generalizes hard problems such as \ProblemName{Maximum Clique} and \ProblemName{Hamiltonicity}~\cite{karp1972},
but unlike many natural NP-hard problems, it requires $N^{\Omega(N)}$ time where $N=\card{V(G)}+\card{V(H)}$ is the total number of vertices, assuming the exponential time hypothesis (ETH)~\cite{CFGKMP16}. 
Hence, most past research addressed its special cases that are in $P$, including the case where the pattern graph is of constant size~\cite{marx14}, or when both graphs are trees~\cite{AbboudWY15}, biconnected outerplanar graphs~\cite{Lingas89}, two-connected series-parallel graphs~\cite{lovasz2009}, and more~\cite{Dessmark00, Matousek92}. 
We will focus on a version called \KTR, where the pattern is a tree $T$ on $k$ nodes. In the directed version of the problem, denoted \DKTR, the edges of $G$ and $T$ are oriented, allowing also anti-parallel edges in $G$\footnote{$T$ need not be an arborescence, only its underlying undirected graph is a tree.}. Throughout, unless accompanied with the word \textit{directed}, \KTR and \NTR refer to their undirected versions.
\DKTR can only be harder than \KTR \,- even when the directed tree $T$ is an arborescence, as one can reduce the undirected version to it with essentially no loss\footnote{This could be done in the following way. Define the host graph $G'$ to be $G$ with edges in both directions, and direct the edges in $T$ away from an arbitrary vertex $v\in T$ to create the directed tree $T'$, which is thus an arborescence. Clearly, the directed instance is a yes-instance if and only if the undirected instance also is.}.
A couple of different techniques were used in order to design algorithms for \DKTR. The color-coding method, designed by Alon, Yuster, and Zwick~\cite{Alon95}, yields an algorithm with running time $O^*((2e)^{k})$. Later, a new method utilized \ProblemName{kMLD} 
(stands for $k$ Multilinear Monomial Detection -- the problem of detecting multilinear monomials of degree $k$ in polynomials presented as circuits) 
to design a \DKTR algorithm with running time $O^*(2^k)$~\cite{Koutis16}. 

%

\paragraph*{Our Results.} 
The first result connects our conjecture to the \DNTR problem (see Figure~\ref{Figures:Reductions}), which is \DKTR with $k=n$.
This problem includes as a special case the well known \DHAM problem, which asks to determine whether a directed graph $G$ contains a simple path (or cycle) that visits all the nodes (the Hamiltonian cycle and path problems are easily reducible to each other with only small overhead). 
Next, we show that an algorithm that breaks Log-SeCoCo implies a fast algorithm for \DNTR.

\begin{thm}\label{thm:main}
Suppose Log-SeCoCo fails, namely, there is $\varepsilon>0$ such that 
for every $\Delta=O(1/\varepsilon\cdot \log n)$, \dSCO 
can be solved in time $O^*(2^{(1-\varepsilon)n})$. 
Then for some $\delta(\varepsilon)>0$, 
\DNTR on $\tilde{n}$ nodes can be solved in time $O^*(2^{(1-\delta)\tilde{n}})$.
This holds even when in \dSCO, every optimal solution is of size $O(\varepsilon n/\log n)$ and consists of disjoint sets.
\end{thm}

In the special case of \DHAM, 
we actually reduce to rather constrained instances of \SCO.

\begin{thm}\label{thm:DHAM}
Suppose Log-SeCoCo fails, namely, there is $\varepsilon>0$ such that for every $\Delta=O(1/\varepsilon\cdot \log n)$, \dSCO 
can be solved in time $O^*(2^{(1-\varepsilon)n})$. 
Then for some $\delta(\varepsilon)>0$,
\DHAM on $\tilde{n}$ nodes can be solved in time $O^*(2^{(1-\delta)\tilde{n}})$. 
This holds even when in \dSCO, all sets are of the same size and every optimal solution is of size $O(\varepsilon n/\log n)$ and consists of disjoint sets.
\end{thm}

We can also show that even moderate improvements to the fastest known running time for \dSCO, 
namely, to the $O^*(2^{(1-1/2\Delta)n})$ time algorithm of~\cite{Koivisto09}, 
implies improvements for \DNTR and for \DHAM (Section~\ref{Section:ModerateImprovements}).

\ifarxiv
\begin{figure}[!ht]
	\centering
		\includegraphics[width=1.0\textwidth,center]{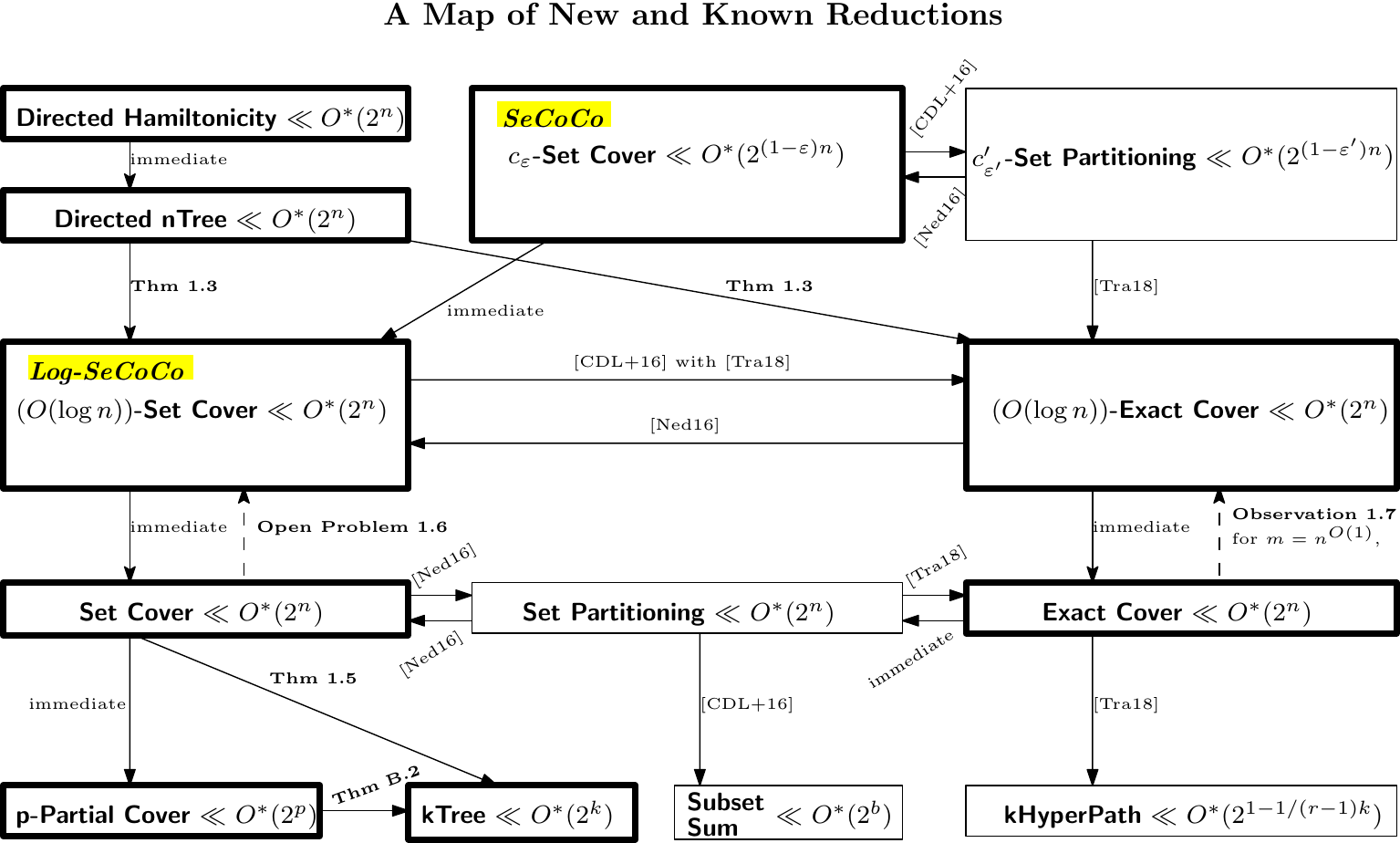}
\else
\begin{figure}[!ht]
	\centering
		\includegraphics[width=1.0\textwidth,center]{Figures/Reductions3}
\fi
   \caption[-]{
An arrow from a box with $A\, \ll \,O^*(2^{n_A})$ to $B \,\ll \,O^*(2^{n_B})$ represents a reduction from problem $A$ to problem $B$, such that if $B$ can be solved in time $O^*(2^{(1-\varepsilon)n_B})$ then $A$ can be solved in time $O^*(2^{(1-\varepsilon'(\varepsilon))n_A})$. 
We denote by $b$ the number of bits required to represent the integers in \ProblemName{Subset Sum}, and by $r$ the uniformity parameter in \ProblemName{kHyperPath}. 
The problems we focus on are drawn in thick frames. 
   }
   \label{Figures:Reductions}

\hrule
\end{figure}

Our next result, whose proof appears in Section~\ref{Section:SCO_KTR}, shows that the $2^k\poly(n)$ running time of \KTR by~\cite{Koutis16} is actually optimal 
(up to exponential improvements) even when considering the undirected version, assuming SeCoCo or even weaker hypotheses such as Log-SeCoCo.
\begin{thm}\label{Thm:SCO_KTR}
If for some fixed $\varepsilon>0$, \KTR can be solved in time $O^*((2-\varepsilon)^k)$, then for some $\delta(\varepsilon)>0$, \SCO on $n$ elements and $m$ sets can be solved in time $O^*((2-\delta)^n)$.
\end{thm}
In fact, our reduction also works from the more general \PPC problem, 
whose input is similar to the \SCO problem but with an additional integer $p$, 
and the goal is to find the smallest sub-collection of sets whose union contains at least $p$ elements (rather than all elements). For simplicity, we first present the reduction from \SCO to \KTR (Section~\ref{Section:SCO_KTR}), and then we show how to adjust it to be from \PPC (Subsection~\ref{Section:PPC}).
%
%
%
%
\paragraph*{Discussion.}
Our first result (Theorem~\ref{thm:main}) supports the validity of Log-SeCoCo
based on the \DNTR problem, which we believe does not admit an $O^*(2^{(1-\varepsilon)n})$-time algorithm, for two reasons.
First, this problem includes the well-known \DHAM problem, 
and in the last $50$ years no algorithm significantly faster than $O^*(2^n)$-time was found for it, 
despite extensive efforts~\cite{bellman1960,held1961,bellman1962,woeg03}
and in contrast to progress on its undirected version~\cite{bjorklund14ham}. 
Second, for a generalization of \NTR and \KTR variants,
namely, for \SI where the pattern is an arbitrary graph of arbitrary size, 
a time lower bound $n^{\Omega(n)}$ is known assuming ETH \cite{CFGKMP16},
even when the host and pattern graphs have the same number of nodes. 
We see it as evidence that also \DKTR does not become easier 
as the size $k$ of the pattern graph increases all the way to $k=n$,
which would imply that the conditional lower bound in Theorem~\ref{Thm:SCO_KTR}
which shows that \KTR cannot be solved in time $O^*(2^{(1-\varepsilon)k})$, 
extends to $k=n$.
%
If true, then by our results, 
solving \SCO significantly faster than $O^*(2^n)$-time is equivalent to achieving the same running time in the special case of \dSCO with $\Delta=O(\log n)$, which can be seen as an analogue to the SETH sparsification lemma~\cite{Impa01spar}. 
Another interesting consequence of our results is that if \KTR can be solved significantly faster than $O^*(2^k)$ than \DNTR can be solved significantly faster than $O^*(2^n)$. Such a reduction from a directed problem to its undirected version is not obvious, even when the latter has extra freedom in the form of parameterization. 
A potentially interesting conclusion from the special instances of \dSCO produced in Theorem~\ref{thm:DHAM}, where the goal could be stated as finding a sub-collection of disjoint sets that covers the entire ground set, which we call \EXC, is that \DHAM could be more closely related to \EXC than to \SCO. This is despite the fact that \SCO and \EXC were shown to be equivalent with respect to solvability in $O^*(2^{(1-\varepsilon)n})$ time~\cite{neder16,Trabelsi18}, as there is an exponential blowup in the number of sets in the reduction from \SCO to \EXC. As we observe below, \EXC with polynomially many sets can indeed be solved significantly faster than $O^*(2^n)$.
%
See Figure~\ref{Figures:Reductions} for an overview of new and known reductions, where problem $A$ being drawn above problem $B$ implies that there is a path, and a reduction, from $A$ to $B$.
The following open problem formalizes the foregoing discussion.

\begin{open}\label{open1}
Does an $O^*(2^{(1-\varepsilon)n})$-time algorithm for \dSCO with $\Delta=O(\log n)$ 
imply an $O^*(2^{(1-\varepsilon'(\varepsilon))n})$-time algorithm for \SCO?
\end{open}

Perhaps surprisingly, 
we can resolve the \EXC analogue of Open Problem~\ref{open1} 
in the special but common case $m=n^{O(1)}$, as follows. Here, $O(c\log n)$-\EXC is \EXC with sets of size bounded by $O(c \log n)$.

\begin{observation}
If for some fixed $\varepsilon>0$ and $c>0$, $O(c\log n)$-\EXC can be solved in time $O^*(2^{(1-\varepsilon)n})$, then for some $\delta(\varepsilon)>0$, \EXC with $m=O(n^c)$ can be solved in time $O^*(2^{(1-\delta)n})$.
\end{observation}
To see this, simply guess which sets of size larger than $\Delta$ 
participate in an optimal solution, 
using an exhaustive search over at most $n\cdot\binom{m}{n/\Delta}$ choices, 
and then apply the assumed algorithm for the remaining sets. 

We note that the results can be easily generalized to weighted \DHAM (i.e., TSP) and \DNTR by using a generalized conjecture about the weighted version of \SCO,
whose input is similar to the \SCO only with a positive weight for each set, and the goal is to find a minimum-weight sub-collection whose union is the entire ground set. The generalized conjecture then states that for every fixed $\varepsilon>0$, weighted \SCO with the cardinality of every set bounded by $O(1/\varepsilon\cdot \log n)$ cannot be solved in time $O^*(2^{(1-\varepsilon)n})$.

\paragraph*{Prior Work.} 
Relevant state-of-the-art algorithms to \SCO and \SI variants are as follows.
\SCO can be solved in time $(m+2^{n})\poly(n)$~\cite{Bjorklund09}, 
which for $m=n^{\omega(1)}$ is faster than the aforementioned $O(mn2^n)$ algorithm of~\cite{fomin04}. 
The case where all sets are of size $q$ and the goal is to determine whether $p$ pairwise-disjoint sets can be packed, can be solved in time $O^*(2^{(1-\varepsilon)pq})$ for $\varepsilon(q)>0$~\cite{bjor17}.
Determining whether a \SCO instance has a solution of size at most $\sigma n$ can be done in time $O^*(2^{(1-\Omega(\sigma^4))n})$~\cite{neder16}.
The fastest known running time for \DHAM is $O^*(2^{n-\Theta(\sqrt{n/\log n})})$
\cite{bjorklundfastest2016}.
Finally, several problems, including \DHAM and \SCO, were shown to belong to the class EPNL, defined as all problems that can be solved by a non-deterministic turing machine with space $n+O(\log n)$ bits~\cite{Iwata2015}.
%

\paragraph*{Techniques.} 
To demonstrate our basic technique for Theorems~\ref{thm:main} and~\ref{thm:DHAM}, let us present an extremely simple reduction from \DHAM to \dSCO with $\Delta=O(\log n)$. 
Given a directed graph $G$, first guess (by exhaustive search) 
a relatively small set of nodes (``representatives''), 
and an ordering for them $z_1,z_2,\ldots$ in a potential Hamiltonian cycle. 
Then construct a \SCO instance whose ground set is the nodes of $G$
and has the following sets: for every possible path of length $\Delta$ in $G$ 
from some $z_i$ to $z_{i+1}$ that does not visit any representative in between,
there is a set that contains all the nodes in this path except for $z_{i+1}$. 
A Hamiltonian cycle in $G$ clearly corresponds 
to a set cover using exactly $n/\Delta$ sets, and vice versa. The main challenge we deal with when reducing from the more general \DNTR is that the pattern tree does not decompose easily into appropriate subgraphs.

The intuition for Theorem~\ref{Thm:SCO_KTR} is as follows. In the reduction from \SCO to \KTR
we first guess a partition of $n$ (the number of elements)
that represents how an optimal solution covers the elements, by exhaustive search over $2^{O(\sqrt{n})}$ unordered partitions of $n$. 
Then, we represent the \SCO instance using a \SI instance,
whose pattern tree $T$ succinctly reflects the guessed partition of $n$, and the idea is that this tree is isomorphic to a subgraph of the \SCO graph 
if and only if the \SCO instance has a solution that agrees with our guess.
The main difficulty here is that we reduce to the undirected version of \KTR, and thus additional attention is required to make the tree fit only in specific locations in the host graph.

\section{Reduction from \DNTR to \SCO}\label{Section:NTR_SCO}

In this section we prove Theorem~\ref{thm:main}.
The heart of the proof is actually the following lemma. 

\begin{lma}\label{lemma:main}
  \DNTR on $\tilde{n}$ nodes can be reduced,
  for every $\Delta\in [\tilde{n}]$,
  to $O(\tilde{n}^{9\tilde{n}/\Delta})$ instances of \dSCO,
  each with $n\leq \tilde{n}+9\tilde{n}/\Delta$ elements,
  in time $O(\tilde{n}^{\Delta+1}+\tilde{n}^{9\tilde{n}/\Delta})$. 
\end{lma}

\begin{proof}[Proof of Theorem~\ref{thm:main}]
  Assume there is an algorithm for \dSCO on $n$ elements and $\Delta=O(1/\varepsilon\cdot \log n)$
  that runs in time $O^*(2^{(1-\varepsilon)n})$. 
  Given an instance of \DNTR on $\tilde{n}$ nodes,
  apply Lemma~\ref{lemma:main} with 
$$\Delta=81/\varepsilon\cdot\log \tilde{n}=O(1/\varepsilon\cdot\log n), $$ 
and then solve each of the resulting 
$O({\tilde{n}}^{9\tilde{n}/\Delta})=O^*(2^{\varepsilon \tilde{n} / 9 })$ 
instances of \dSCO, using the assumed algorithm, in time  $$O^*(2^{(1-\varepsilon)n})\leq O^*(2^{(1-\varepsilon)(\tilde{n}+9 \tilde{n}/(81/\varepsilon \cdot\log \tilde{n}))})\leq O^*(2^{(1-\varepsilon)(\tilde{n}+ \varepsilon \tilde{n}/(9 \cdot\log \tilde{n}))}). $$
The total running time is
$$O^*(2^{81/\varepsilon\cdot \log^2 \tilde{n}+\log \tilde{n} +\varepsilon \tilde{n}/9 +(1-\varepsilon)(\tilde{n}+ \varepsilon \tilde{n}/(9\cdot\log \tilde{n}))})\leq O^*(2^{\tilde{n}-\varepsilon \tilde{n}/2}), $$
which concludes the proof for $\delta(\varepsilon)=\varepsilon/2$.
\end{proof}

It remains to prove Lemma~\ref{lemma:main}, and we start with an overview of this proof. 
Consider an instance $(G,T)$ of \DNTR, and for this overview, 
assume that the tree $T$ is rooted at some node $r$, 
and all edges are directed away from it. 
The idea is to create roughly $\tilde{n}^{9\tilde{n}/\Delta}$
instances of \dSCO on $n\leq \tilde{n}+9\tilde{n}/\Delta$ elements each,
such that at least one of them has a solution of size $t\leq 9\tilde{n}/\Delta$
if and only if the instance $(G,T)$ has a solution. 
The first step is to cover the tree $T$ with $t$ small subtrees, 
each of size at most $\Delta$, such that the union of their node sets is $T$ 
and they may intersect only at their roots
(the root of a subtree is the node closest to $r$).
Then guess, by enumerating over all possible choices, 
how the solution to $(G,T)$ maps the root of each subtree to a node in $G$,
and create a corresponding an instance of \dSCO.
For every such instance, perform an inner enumeration to further guess,
what is the (unordered) set of nodes in $G$ that each subtree is mapped to, 
and add a corresponding set to the \dSCO instance, 
but only if this guess does not violate the local and global structure of $T$. That is, taking into account the edges within and between the subtrees,
by testing whether the set can be an isomorphic copy of the subtree,
testing for the edges between roots, respectively.
For the correctness, we need to show that
a solution of size $t$ to the \dSCO instance implies
a one-to-one correspondence between the $t$ sets and the roots of the subtrees,
and hence a copy of $T$ in $G$. 
%
The general case where the edges of $T$ are orientated arbitrarily is similar, 
except that the edge orientations are taken into account when comparing subtrees
but not when computing a cover of $T$ by small subtrees.

We proceed to the algorithm that computes the aforementioned cover of $T$
by small subtrees. 
This algorithm traverses the tree using DFS
and add subtrees to the cover whenever the DFS accumulates enough nodes,
see Algorithm~\ref{alg:partition} for full details. 
Its output is a set $S$,
where each $s\in S$ is a connected subset of the nodes of $T$,
and thus we can refer to each such $s$ as a subtree of $T$,
and let $r(s)$ denote its root, i.e., its node that is closest to $r$ in $T$. 
The following lemma describes the guarantees of this algorithm
and will be later used to prove Lemma~\ref{lemma:main}.
\begin{lma} \label{Lemma:Aux}
  Given a tree $T$ with root $r$ on $\tilde{n}$ nodes
  and an integer $l\leq \tilde n$, 
Algorithm~\ref{alg:partition} finds in polynomial time a collection $S$ of subtrees of $T$ such that:
\begin{enumerate} 
\renewcommand{\theenumi}{\alph{enumi}}
\item \label{it:a}
  the number of nodes in each subtree is at most $2(l-1)$;
\item \label{it:b}
  every node in $T$ is in some subtree; 
\item \label{it:c}
  two subtrees in $S$ may only intersect in their roots; and
\item \label{it:d}
  the number of subtrees is $\card{S} \leq \frac{3\tilde{n}}{l-1}$.
\end{enumerate}
\end{lma}

\begin{algorithm}
  \caption{}\label{alg:partition}
  \begin{algorithmic}[1]
    \Require {tree $T$ rooted at $r$ and size parameter $l\in [n]$}
    \Ensure {cover $S$ of $T$ by subtrees of size at most $2(l-1)$}
    \State {$S\leftarrow\emptyset$}
    \ForAll{$v\in V$}
      $s(u)\leftarrow\{u\}$ 
      \EndFor
    \State {traverse $T$ using a DFS from $r$, and whenever returning from a node $v$ to its parent $p$ in $T$, do the following:}
	    \Indent
    	\State{let $s(p)\leftarrow s(p)\cup s(v)$}
		\If{$\card{s(p)}\geq l$}
	    	\State{add $s(p)$ to $S$}
\label{alg:partition:addbig}
	    	\If{$p$ has unvisited children}
	    		\State {let $s(p)\leftarrow \{p\}$}
    	 	\Else {} $\text{let } s(p)\leftarrow \emptyset$
	    	\EndIf
	    \ElsIf {$p$ has no unvisited children and $p\in s$ for some $s\in S$} 
	    	\State {add $s(p)$ to $S$ and let $s(p)\leftarrow\emptyset$}\label{alg:partition:addsmall}
    	\ElsIf {$p$ is the last node traversed in the tree} 
	    	\State {add $s(p)$ to $S$}\label{alg:partition:addlastsmall}
	    \EndIf
	\EndIndent
	\State {\textbf{return} $S$}
  \end{algorithmic}
\end{algorithm}

\begin{proof}[Proof of Lemma~\ref{Lemma:Aux}]
We first show that items \eqref{it:a}--\eqref{it:c} are satisfied by the output of Algorithm~\ref{alg:partition}.
Since in the worst case Algorithm~\ref{alg:partition} adds a subtree in the first time the accumulated number of nodes exceeds $l$, the number of nodes of each subtree is bounded by $2(l-1)$.
In addition, every node $v$ appears in some subtree,
since at some point during the DFS it will be the child, 
and then it will be passed up the tree and eventually added to $S$.
To see why the last requirement holds,
observe that whenever an accumulated set is passed up the tree and encounters an existing root, this set will be added to $S$.

To prove item \eqref{it:d}, denote
denote by $\Sbig$ the collection of sets in $S$ of size at least $l$
(added in line~\ref{alg:partition:addbig}),
and by $\Ssml$ the collection of sets in $S$ of size smaller than $l$ (added in lines~\ref{alg:partition:addsmall} and~\ref{alg:partition:addlastsmall}).
A set $s\in\Ssml$ was created only if $r(s)$ at the time of its creation
was the root of at least one (other) set in $\Sbig$ (line~\ref{alg:partition:addsmall}) or was the last traversed node in the DFS (line~\ref{alg:partition:addlastsmall}).
Together with the fact that each root has at most one set from $\Ssml$,
we conclude that each set $s\in \Ssml$ excluding at most one,
can be associated with a distinct set in $\Sbig$, one that contains $r(s)$.
Hence, $\card{\Ssml}-1 \leq \card{\Sbig}$. 
The big sets have size at least $l$,
and except for their roots they have distinct vertices,
hence $\card{\Sbig}\leq \frac{\tilde{n}}{l-1}$.
We conclude that
$$
\card{S}=\card{\Ssml}+\card{\Sbig}\leq 2\card{\Sbig}+1\leq 
\frac{2\tilde{n}}{l-1}+1\leq \frac{3\tilde{n}}{l-1},
$$
which completes the proof of Lemma~\ref{Lemma:Aux}. 
\end{proof}

\begin{proof}[Proof of Lemma~\ref{lemma:main}]
We describe the reduction in stages.
\begin{itemize}
\item Apply the aforementioned Algorithm~\ref{alg:partition} for partition $T$ into subtrees that satisfy the conditions in Lemma~\ref{Lemma:Aux}. By picking $l=\Delta/3+1$, we obtain that each set is bounded by $\Delta$ and that $\card{S}\leq 9\tilde{n}/\Delta$. Hence, the cardinality of $R:=\{r(s)\}_{s\in S}$ is bounded by $9\tilde{n}/\Delta$. For $S$ returned by Algorithm~\ref{alg:partition}, let $R_T=\{r(s): s\in S\}$ (note that $\card{R_T}$ may be smaller than $\card{S}$).
\item Then, guess $\card{R_T}$ nodes in $G$ that will function as the image of the nodes in $R_T$ in a potential subgraph isomorphism function and denote them by $R_G$, and then guess a bijection $f$ from $R_T$ to $R_G$. The guessing is done by exhaustive search over $\binom{\tilde{n}}{\card{R_T}}$ choices of nodes, and together with the number of ways to choose a bijection it can be done in time $\binom{\tilde{n}}{\card{R_T}}\card{R_T}!$.
\item Finally, enumerate all sets $s'$ of nodes of size at most $\Delta$ in $G$, and denote by $G(s')$ the graph induced from each on $G$. For every subtree $s\in S$, look by brute force for an isomorphic copy of $s$ in subgraphs $G(s')$ that contain $f(r(s))$ as a root and no other node in $R_G$, and that satisfy $\card{s'}=\card{s}$. 
For each one that was found, add to the constructed \SCO instance a set ${s'}_G$ with the root $r'$ labeled ${r'}_{s}$ where $s$ corresponds to the subtree $s$ of $T$ whose copy found to be in $G(s')$. 
Note that the number of elements in the \SCO instance is exactly $\tilde{n}-\card{R_T}+\card{S}$, and that the time spent per each subgraph isomorphism test is at most $\card{s}!\leq \Delta!$, and thus the total time spent in this step is $\card{S} \binom{\tilde{n}}{\Delta} \Delta!=\card{S} \tilde{n}\cdot (\tilde{n}-1) \cdot\cdot\cdot (\tilde{n}-\Delta+1)
\leq 9\tilde{n}/\Delta\cdot \tilde{n}^{\Delta}\leq \tilde{n}^{\Delta+1}$.

\end{itemize}

Now we show that the size constraints follow. As $\card{R_T}\leq 9\tilde{n}/\Delta$, similar to before, the number of \SCO instances is bounded by
$$
\binom{\tilde{n}}{9\tilde{n}/\Delta}(9\tilde{n}/\Delta)! =
\tilde{n}\cdot(\tilde{n}-1)\cdot\cdot\cdot (\tilde{n}-9\tilde{n}/\Delta+1)\leq
\tilde{n}^{9\tilde{n}/\Delta}
$$
as required.

We now prove that at least one of the \SCO instances has solution of size at most $\card{S}$ (in fact exactly $\card{S}$ as no smaller solutions available) if and only if the \DNTR instance is a yes instance. For the first direction, assume that the \DNTR instance is a yes instance. Considering the isomorphic copy of $T$ in $G$, its $\card{S}$ subtrees as Algorithm~\ref{alg:partition} outputs on $T$ will be sets in the \SCO instance the reduction outputs, and so it has a solution of size at most $\card{S}$. For the second direction, if a \SCO instance has a solution $I$ of size at most $\card{S}$ and since the number of labeled roots is $\card{S}$, it must be that for each subtree $s\in S$ its labeled root is in exactly one set in $I$, and so $\card{I}=\card{S}$. Since $I$ is a legal solution and $S$ covers all the nodes, no node in $V(G)\setminus R_G$ appears twice in $I$. The conclusion is that these sets together form the required tree, concluding the proof of Lemma~\ref{lemma:main}.
\end{proof}

We note that in the case of Theorem~\ref{thm:DHAM} for \DHAM, we do not have to use Algorithm~\ref{alg:partition}, but simply guess $n/\Delta$ representative nodes in $G$ and their ordering in the potential cycle, and then enumerate all paths of size $\Delta$ to represent paths between consecutive representatives. Hence we obtain a \dSCO instance with the additional constraints of Theorem~\ref{thm:DHAM}.

\section{Reduction from \SCO to \KTR}\label{Section:SCO_KTR}

In this section we prove Theorem~\ref{Thm:SCO_KTR}. In order to make the proof simpler, we will have an assumption regarding the \SCO instance, as follows. For a constant $g>0$ to be determined later, 
we can assume that all the sets in the \SCO instance are of size at most $n/g^2$, 
as otherwise such instance can already be solved significantly faster than $O^*(2^n)$, proving the theorem in a degenerate manner. 
We formalize it as follows.

\begin{assumption}\label{asm1}
\textit{All the sets in the \SCO instance are of size at most $n/g^2$.}
\end{assumption}
To justify this assumption, notice that one can remove all sets of size more than $n/g^2$ from the \SCO instance. Indeed, if some optimal solution for the \SCO instance contains a set of size at least $n/g^2$, such optimal solution can be found by simply guessing one set of at least this size (using exhaustive search over at most $m$ choices) and then applying the known dynamic programming algorithm on the still uncovered elements (at most $n-n/g^2$ of them), and return the optimal solution in total time $O^*(2^{(1-1/g^2)n})$.
We continue to the following lemma, which is the heart of the proof.

\begin{lma}\label{LemmaKTR}
For every fixed $\varepsilon>0$, \SCO on a ground set $N=[n]$ and a collection $M$ of $m$ sets that satisfies assumption~\ref{asm1}, 
can be reduced to $2^{O(\sqrt{n})}$ instances of \KTR with $k=(1+\varepsilon)n+O(1)$.
\end{lma}
We will use this lemma to prove Theorem ~\ref{Thm:SCO_KTR}, the proof of Lemma~\ref{LemmaKTR} will be given after.

\begin{reptheorem}{Thm:SCO_KTR}[\textit{restated}]
If for some fixed $\varepsilon>0$, \KTR can be solved in time $O^*((2-\varepsilon)^k)$, then for some $\delta(\varepsilon)>0$, \SCO on $n$ elements and $m$ sets can be solved in time $O^*((2-\delta)^n)$.
\end{reptheorem}

\begin{proof}[Proof of Theorem~\ref{Thm:SCO_KTR}]
Assume that for some $\varepsilon'\in (0,1)$, \KTR can be solved in time $O^*((2-\varepsilon')^k)\leq O^*(2^{(1-\varepsilon'/2)k})$. We reduce the \SCO instance by applying Lemma~\ref{LemmaKTR} 
with $\varepsilon=\varepsilon'/4$, 
and then solve each of the $2^{c_1\sqrt{n}}$ instances of \KTR 
in the assumed time of $O^*(2^{(1-\varepsilon'/2)((1+\varepsilon)n+c_2)})$, 
where $c_1,c_2>0$ are the constants implicit in the terms 
$2^{O(\sqrt{n})}$ and $O(1)$ in the lemma, respectively.
The total running time is 
$ O^*(2^{(1-\varepsilon'/2)(1+\varepsilon)n+c_1\sqrt{n}})
  = O^*(2^{(1-\varepsilon'/4-\varepsilon'^2/8)n+c_1\sqrt{n}})
  \leq O^*(2^{(1-\varepsilon'/4)n})\leq O^*((2-\varepsilon'/4)^n)$, 
which concludes the proof for $\delta(\varepsilon')=\varepsilon'/4$.
\end{proof}

To outline the proof of Lemma~\ref{LemmaKTR}, 
we will need the following definition. 
For an integer $a>0$, let $p(a)$ be the set of all unordered partitions of $a$, where a \emph{partition} of $a$ is a way of writing $a$ as a sum of positive integers, 
and \emph{unordered} means that the order of the summands is insignificant. 
The asymptotic behaviour of $\card{p(a)}$ (as $a$ tends to infinity) 
is known~\cite{hardy1918} to be
$$
  e^{\pi \sqrt{{2a}/{3}}}/(4a\sqrt{3})=2^{O(\sqrt{a})}. 
$$
It is possible to enumerate all the partitions of $a$ with constant delay between two consecutive partitions, 
exclusive of the output~\cite[Chapter 9]{nijen78}.

Now the intuition for our reduction of \SCO to \KTR
is to first guess a partition of $n$ (the number of elements)
that represents how an optimal solution covers the elements, as follows. 
Associate each element arbitrarily with one of the sets that contain it
(so in effect, we assume each element is covered only once) 
and count how many elements are covered by each set in the optimal solution.
This guessing is done by exhaustive search over $p(n)\le 2^{O(\sqrt{n})}$ 
partitions of $n$. 
Then, we represent the \SCO instance using a \SI instance,
whose pattern tree $T$ succinctly reflects the guessed partition of $n$.
The idea is that the tree is isomorphic to a subgraph of the \SCO graph 
if and only if the \SCO instance has a solution that agrees with our guess.

\begin{replemma}{LemmaKTR}[\textit{restated}]
For every fixed $\varepsilon>0$, \SCO on a ground set $N=[n]$ and a collection $M$ of $m$ sets that satisfies assumptions~\ref{asm1}, 
can be reduced to $2^{O(\sqrt{n})}$ instances of \KTR with $k=(1+\varepsilon)n+O(1)$.
\end{replemma}

\begin{proof}[Proof of Lemma~\ref{LemmaKTR}]
Given a \SCO instance on $n$ elements $N=\{n_i: i \in [n]\}$ and $m$ sets $M=\{S_i\}_{i\in [m]}$ and an $\varepsilon>0$, construct $2^{O(\sqrt{n})}$ instances of \KTR as follows. For a constant $g(\varepsilon)$ to be determined later, the host graph $G_g=(V_g,E_g)$ is the same for all the instances, and is built on the bipartite graph representation of the \SCO instance, with some additions. This is done in a way that a constructed tree will fit in $G_g$ if and only if the \SCO instance has a solution that corresponds to the structure of the tree, as follows (see Figure~\ref{figs:G}). The set of nodes is $V_g=N\cup M\cup M_g\cup R \cup \{r_{g}, r_1, r_2, r\}$, where 
$M_g=\{ X\subseteq M : \card{X}=g \}$ 
 and $R=\{v_j^i:i\in [4], j\in [n/(g/2)]\}$.
Intuitively, the role of $M_g$ is to keep the size of the trees small by representing multiple vertices in $M$ (multiple sets in \SCO) at once as the "powering" technique for \SCO done in~\cite{cygan16}\footnote{We can slightly simplify this step in the construction by using the equivalence from~\cite{cygan16} between solving \SCO in time $O^*(2^{(1-\varepsilon)n})$ and in time $O^*(2^{(1-\varepsilon')(n+t)})$ where $t$ is the solution size. However, we opted to reduce directly from \SCO for compatibility with our parameters and for sake of generality.}, 
and the role of $R$ and $\{r_{g}, r_1, r_2, r\}$ is to enforce that the trees the reduction
constructs will fit only in certain ways.  

The set of edges is constructed as follows. 
Edges between $N$ and $M$ are the usual bipartite graph representation of \SCO 
(i.e., connect vertices $n_j\in N$ and $S_i\in M$ whenever $n_j\in S_i$).
Also, connect vertex $X\in M_g$ to vertex $n_j\in N$ 
if at least one of the sets in $X$ contains $n_j$. 
Additionally, add edges between $r_{g}$ and every vertex in $M_g$, and $v^4_j\in R$ for $j\in [n/(g/2)]$, between $r_i$ and $v^i_j$ for every $i\in \{1,2\}$ and $j\in [n/(g/2)]$, and finally between $r$ and every vertex $v\in \{r_{g}, r_1, r_2\}$, $S_i\in M$, and $v^3_j\in R$ for $j\in [n/(g/2)]$.

\ifarxiv
\begin{figure}[!ht]
	\centering
		\includegraphics[width=1.0\textwidth,left]{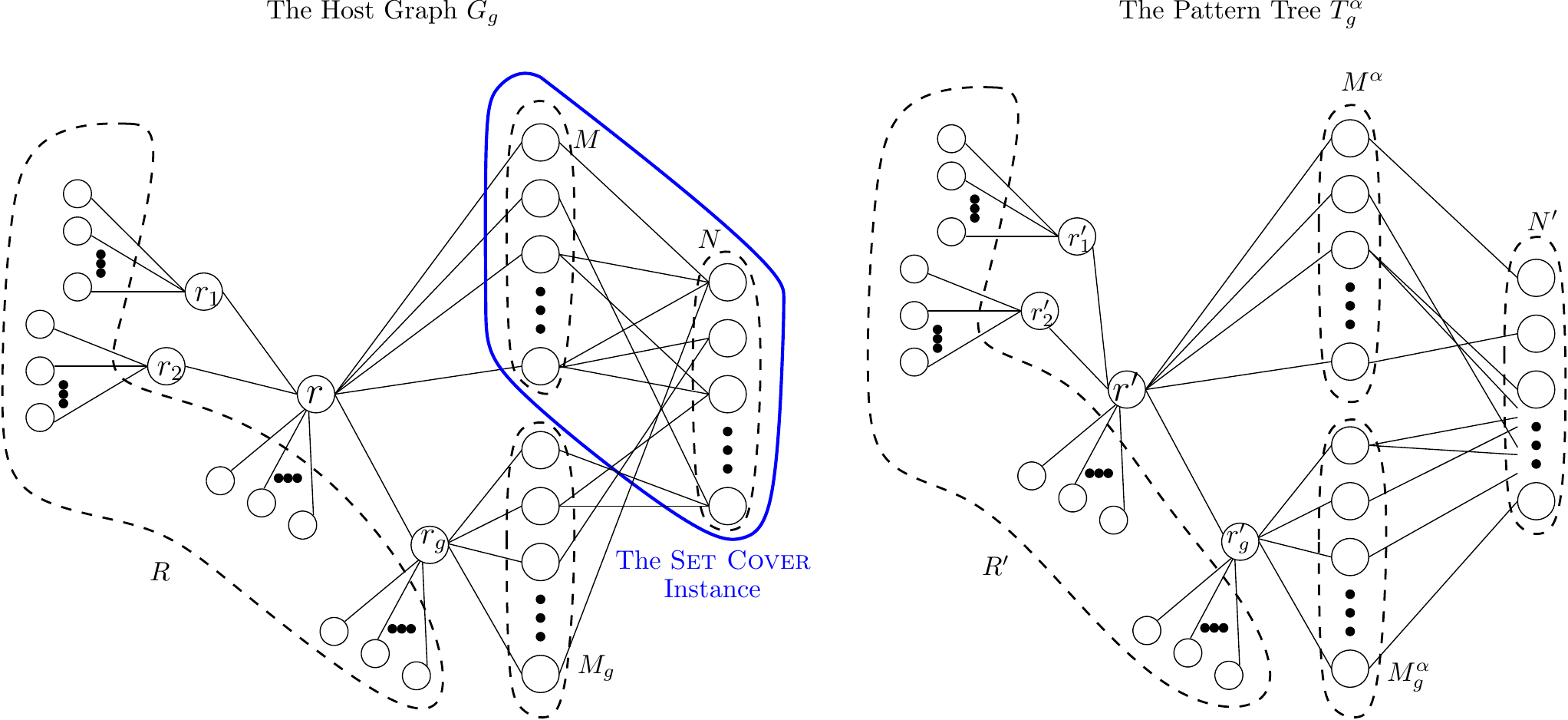}
\else
\begin{figure}[!ht]
	\centering
		\includegraphics[width=1.0\textwidth,left]{figures/G+T}
\fi
   \caption[-]{An illustration of part of the reduction. The \SCO instance is depicted in blue, and sets of vertices are indicated by dashed curves.
   }
   \label{figs:G}
\end{figure}

Next, construct $2^{O(\sqrt{n})}$ trees such that identifying those that are isomorphic to a subgraph of $G_g$ will determine the optimum of the \SCO instance. 

For every partition $\alpha=(p_1,p_2,...,p_l)\in p(n)$ (with possible repetitions) where $p(n)$ is as defined above, construct a tree $T_g^{\alpha}=(V_g^{\alpha}, E_g^{\alpha})$. This tree has the same set of edges and vertices as $G_g$, except for the vertices in $M\cup M_{g}$ and the edges incident to them, which are replaced by a set of new vertices $M^{\alpha}\cup M_g^{\alpha}$, and these new vertices are connected to the rest in a way that the resulting graph is a tree. 
In more detail, $V_g^{\alpha}=N'\cup M^{\alpha}\cup M_g^{\alpha}\cup R'\cup \{r'_g, r'_1, r'_2, r'\}$ where $N', R', r'_g, r_1', r_2', r'$ are tagged copies of the originals, and $M^{\alpha}, M_g^{\alpha}$ are initialized to be $\emptyset$.

We define $\alpha_{g}$ to be a partition of $n$ which is also a shrinked representation of $\alpha$ by partitioning $\alpha$ into sums of $g$ numbers for a total of $\lfloor l/g \rfloor$ such sums, and a remaining of less than $g$ numbers. Formally,
$$\alpha_{g}=(\sum^g_{i=1} p_i, \sum^{2g}_{i=g+1} p_i,..., \sum^{g \lfloor l/g \rfloor}_{i=(g-1)\cdot \lfloor l/g \rfloor+1} p_i, p_{g \lfloor l/g \rfloor+1},...,p_l)
$$
Note that all the numbers in $\alpha_{g}$ are a sum of $g$ numbers in $\alpha$, except (maybe) for the last $g':=l-g\lfloor l/g \rfloor<g$ numbers in $\alpha_{g}$, a (multi)set which we denote $s(\alpha_{g})$. For every $i\in \alpha_{g}$ (with possible repetitions), add a star on $i+1$ vertices to the constructed tree $T^{\alpha}_g$. If $i\in \alpha_{g}\setminus s(\alpha_g)$, add the center vertex to $M^{\alpha}_g$, connect it to $r'_g$, and add the rest $i$ vertices to $N'$. Else, if $i\in s(\alpha_{g})$, add the center vertex to $M^{\alpha}$, connect it to $r'$, and again add the rest $i$ vertices to $N'$. Return the minimum cardinality of $\alpha$ for which $(G_g, T_g^{\alpha})$ is a yes-instance. To see that this construction is small enough, note that the size of $G_g$ is at most $4+4\cdot n/(g/2)+m^g+m+n$ which is polynomial in $m$, and the size of the tree $T_g^{\alpha}$ is at most
$$
4+4\cdot n/(g/2)+n/g+g+n = n\cdot (1+9/g)+O(1) = n\cdot (1+\varepsilon)+O(1)
$$
where the last equality holds for $g=9/\varepsilon$, and so the size constraint follows.

We now prove that at least one of the trees $T_g^{\alpha}$ returns yes and satisfies $\card{\alpha}\leq d$, if and only if the \SCO instance has a solution of size at most $d$. For the first direction, assume that the \SCO instance has a solution $I$ with $\card{I}\leq d$. Consider a partition $\alpha_I\in p(n)$ of $n$ that corresponds to $I$ in the following way. Associate every element with exactly one of the sets in $I$ that contains it, and then consider the list of sizes of the sets in $I$ according to this association (eliminating zeroes). Clearly, $(G_g, T_g^{{\alpha_I}})$ is a yes-instance and so the reduction will return a number that is at most $\card{I}$.

For the second direction, assume that every solution to the \SCO instance is of size at least $d+1$. We need to prove that for every tree $T^{\alpha}_g$ with $\card{\alpha}\leq d$, $(G_g,T_g^{\alpha})$ is a no-instance. Assume for the contrary that there exists such $\alpha$ for which $(G_g,T_g^{\alpha})$ is a yes-instance with the isomorphism function $f$ from $T_g^{\alpha}$ to $G_g$. We will show that the only way $f$ is feasible is if $f(r')=r$, $f(M^{\alpha})\subseteq M$, $f(M^{\alpha}_g)\subseteq M_g$, and also $f(N')= N$, which together allows us to extract a corresponding solution for the \SCO instance, leading to a contradiction. 
We start with the vertex $r'\in T_g^{\alpha}$. Since its degree is at least $n/(g/2)+3$ and by Assumption~\ref{asm1} and the construction of $G_g$, it holds that $f(r')\notin \{ r_1, r_2 \}\cup R\cup M\cup M_g $. Moreover, if it was the case that $f(r')\in \{ r_g \}\cup N$ then $\{ f(r'_1), f(r'_2)\}\cap (M\cup M_g)\neq \emptyset$, however, the degree of $r'_1$ and $r'_2$ in $T_g^{\alpha}$ is $n/(g/2)$, and the degree of the vertices in $M\cup M_g$ in $G_g$ is at most $g\cdot n/g^2=n/g$, so it must be that $f(r)=r$. 
Our next claim is that $f(r'_g)=r_g$. Observe that Assumption~\ref{asm1} implies that every solution for the \SCO instance is of size at least $g^2$ and so 
$M^{\alpha}_g\neq \emptyset$, which means $r'_g$ in the tree has vertices in distance $2$ from it and away from $r'$, a structural constraint that cannot be satisfied by any vertex in $\{ r_1,r_2 \}\cup R$. Furthermore, the degree of $r'_g$ is at least $n/(g/2)$ and so again by Assumption~\ref{asm1} it is also impossible that $f(r'_g)\in M^{\alpha}$, and hence it must be that $f(r'_g)=r_g$. Finally, by the same Assumption and the degrees of $r_1$ and $r_2$, $f(r'_1)$ and $f(r'_2)$ must be in $\{ r_1,r_2 \}$. Altogether, it must be that $f(M^{\alpha}_g)\subseteq M_g$, $f(M^{\alpha})\subseteq M$ and that $f(N')= N$, and therefore it is possible to extract a feasible solution to the \SCO instance that has at most $d$ sets in it, which is a contradiction, concluding the proof of Lemma~\ref{LemmaKTR}. 
\subsection{Reduction from \PPC}\label{Section:PPC}
In this subsection we show that Theorem~\ref{Thm:SCO_KTR} is correct also assuming a weaker conjecture, that \PPC cannot be solved significantly faster than $O^*(2^p)$.
Notice that \PPC can be solved in time $O^*(2^p)$ 
by a simple application of the method in~\cite{Koutis16},
as pointed out to us by Cornelius Brand and anonymous referees. 
We now reduce from \PPC to \DKTR by following Lemma~\ref{LemmaKTR} 
with the following adjustments.

Instead of enumerating over all the partitions of $n$, do it only for $p$ and hence the number of partitions is $2^{O(\sqrt{p})}$ with each partition $\alpha$ inducing a tree $T^{\alpha}_g$ in a similar way to Lemma~\ref{LemmaKTR}, of size at most $2p/g+p$. Note that Assumption~\ref{asm1} adjusted to the \PPC case hold also here, since it is possible to use the $O^*(2^p)$-time algorithm for \PPC mentioned above after removing large sets of size $\geq p/g^2$.
From here onwards, the proof of correctness is similar to Lemma~\ref{LemmaKTR}, and thus we omit it. 
Regarding running time, assume that for some $\varepsilon'\in (0,1)$, \KTR can be solved in time $O^*((2-\varepsilon')^k)\leq O^*(2^{(1-\varepsilon'/2)k})$.
Setting $g=8/\varepsilon'$ for $\varepsilon'=64(1-\log_{2}(2-\varepsilon))$ (without loss of generality, assume that $\varepsilon'$ is small enough), we get a total running time of 

\vspace{3mm}
$
O(m^{c_1 g}2^{p-p/{g^2}}+2^{(p+2p/g)(1-\varepsilon'/2)+c_2\sqrt{p}}\cdot m^{c_3 g})
$
\begin{align*}
\hspace{35mm} &= O(m^{c_1 8/\varepsilon'}2^{p-\varepsilon'p/64}+ 2^{p+\varepsilon'/4 \cdot p - \varepsilon'/2\cdot p - \varepsilon'^2/8\cdot p + c_1\sqrt{p}} \cdot m^{c_3 4/\varepsilon'})\\
&\leq O(2^{(1-\varepsilon'/64)p} \cdot m^{c_1 4/\varepsilon'})\\
&\leq O((2-\varepsilon)^p \cdot m^{c_1 4/\varepsilon}),
\end{align*}

where $c_1$ is the constant derived from the method of~\cite{Koutis16}, $c_2$ is the constant implicit in the term $2^{O(\sqrt{p})}$, and $c_3$ is the constant in the exponent of $m$ implicit in the term $O^*(2^{(1-)p})$, as required.

\begin{lma}\label{Lemma:PPC_KTR}
For every fixed $\varepsilon>0$, \PPC on a ground set $N=[n]$ and a collection $M$ of $m$ sets can be reduced to $2^{O(\sqrt{p})}$ instances of \KTR with $k=(1+\varepsilon)p+O(1)$.
\end{lma}
We thus proved the following theorem.
\begin{thm}\label{Thm:PPC_KTR}
If for some fixed $\varepsilon>0$, \KTR can be solved in time $O^*((2-\varepsilon)^k)$, then for some $\delta(\varepsilon)>0$, \PPC on $n$ elements and $m$ sets can be solved in time $O^*((2-\delta)^p)$.
\end{thm}




\section{Moderate Improvements to \dSCO Imply New Algorithms for \DNTR and \DHAM}\label{Section:ModerateImprovements}
In this section we show how moderate improvements for variants of \SCO imply new algorithms for \DNTR. Given any algorithm for \dSCO with runtime $f(n,m,\Delta)$, by Lemma~\ref{lemma:main} \DNTR admits an algorithm with running time $O(\tilde{n}^{\Delta}+\tilde{n}^{\tilde{n}/\Delta}f(n,m,\Delta))$. We now demonstrate how this algorithm behaves with different regimes of $\Delta$.

If there exists $\varepsilon>0$ such that for every $\Delta=\poly(\log n)$, $f(n,m,\Delta)=O^*(2^{(1-1/\Delta^{1-\varepsilon})n})$ then by considering $\Delta=\log^{(1+\varepsilon')/\varepsilon} n = \poly(\log n)$ for $\varepsilon'>0$, \DNTR has an algorithm with runtime 
$$
O(2^{\log^{(1+\varepsilon')/\varepsilon+1} \tilde{n}})+O^*(2^{\tilde{n}/\log^{(1+\varepsilon')/\varepsilon-1} \tilde{n}}\cdot 2^{(1-1/(\log^{(1+\varepsilon')/\varepsilon} \tilde{n})^{1-\varepsilon})\tilde{n}})=O^*(2^{(1-1/(\log^{(1+\varepsilon')/\varepsilon-2}\tilde{n}))\tilde{n}})
$$

Considering larger regimes, if for some fixed $\varepsilon>0$, $\delta\in(0,1/2)$, and $\Delta=O({n}^{\delta})$, $f(n,m,\Delta)=O^*(2^{(1-\frac{(1+\varepsilon)\log \Delta}{\delta \Delta})n})$ then \DNTR can be solved in time 
$$
2^{\tilde{n}^{\delta}\log \tilde{n}}+ 2^{\tilde{n}^{1-\delta}\log \tilde{n}}\cdot O^*(2^{(1-\frac{(1+\varepsilon)\log \Delta}{\delta \Delta})\tilde{n}})=
O^*(2^{(1-\varepsilon/\tilde{n}^{\delta})\tilde{n}})=2^{\tilde{n}-\Theta(\tilde{n}^{1-\delta})}
$$
Note that to break the fastest known $2^{\tilde{n}-\Theta(\sqrt{\tilde{n}/\log \tilde{n}})}$ algorithm for \DHAM by~\cite{bjorklundfastest2016}, it is enough to have either $f(n,m,\Delta)=O^*(2^{(1-\frac{(2+\varepsilon)\log \Delta}{\Delta})n})$ for $\Delta=n^{1/2-\delta'}$ with every fixed $\delta'>0$, or $f(n,m,\Delta)=O(m\cdot 2^{(1-\frac{(4+\varepsilon)\log\Delta}{\Delta})n})$ for $\Delta=\sqrt{n}$, taking into account that most algorithms for variants of \SCO that have the factor $m$ in their runtime, do not have it with higher power than one.

\end{proof}

\bibliographystyle{plainurl}

\bibliography{robi}

\begin{thebibliography}{10}

\bibitem{abboud2017seth}
Amir Abboud, Karl Bringmann, Danny Hermelin, and Dvir Shabtay.
\newblock Seth-based lower bounds for subset sum and bicriteria path.
\newblock In {\em 30th Annual ACM-SIAM Symposium on Discrete Algorithms}, SODA
  '19, pages 41--57, 2019.
\newblock \href {http://dx.doi.org/10.1137/1.9781611975482.3}
  {\path{doi:10.1137/1.9781611975482.3}}.

\bibitem{AbboudWY15}
Amir Abboud, Virginia {Vassilevska-Williams}, and Huacheng Yu.
\newblock Matching triangles and basing hardness on an extremely popular
  conjecture.
\newblock In {\em Proceedings of the Forty-seventh Annual ACM Symposium on
  Theory of Computing}, STOC '15, pages 41--50. ACM, 2015.
\newblock \href {http://dx.doi.org/10.1145/2746539.2746594}
  {\path{doi:10.1145/2746539.2746594}}.

\bibitem{Alon95}
Noga Alon, Raphael Yuster, and Uri Zwick.
\newblock Color-coding.
\newblock {\em J. ACM}, 42(4):844--856, July 1995.
\newblock \href {http://dx.doi.org/10.1145/210332.210337}
  {\path{doi:10.1145/210332.210337}}.

\bibitem{bellman1960}
Richard Bellman.
\newblock Combinatorial processes and dynamic programming.
\newblock In {\em Combinatorial analysis}, Proceedings of Symposia in Applied
  Mathematics, pages 217--249. American Mathematical Society, 1960.
\newblock \href {http://dx.doi.org/10.1090/psapm/010}
  {\path{doi:10.1090/psapm/010}}.

\bibitem{bellman1962}
Richard Bellman.
\newblock Dynamic programming treatment of the travelling salesman problem.
\newblock {\em J. ACM}, 9(1):61--63, 1962.
\newblock \href {http://dx.doi.org/10.1145/321105.321111}
  {\path{doi:10.1145/321105.321111}}.

\bibitem{bjorklund14ham}
Andreas Bjorklund.
\newblock Determinant sums for undirected hamiltonicity.
\newblock {\em SIAM Journal on Computing}, 43(1):280--299, 2014.
\newblock \href {http://dx.doi.org/10.1137/110839229}
  {\path{doi:10.1137/110839229}}.

\bibitem{bjorklundfastest2016}
Andreas Bj{\"o}rklund.
\newblock {Below All Subsets for Some Permutational Counting Problems }.
\newblock In {\em 15th Scandinavian Symposium and Workshops on Algorithm Theory
  (SWAT 2016)}, volume~53 of {\em Leibniz International Proceedings in
  Informatics (LIPIcs)}, pages 17:1--17:11, 2016.
\newblock \href {http://dx.doi.org/10.4230/LIPIcs.SWAT.2016.17}
  {\path{doi:10.4230/LIPIcs.SWAT.2016.17}}.

\bibitem{bjor15}
Andreas Bj\"{o}rklund, Dell Holger, and Thore Husfeldt.
\newblock The parity of set systems under random restrictions with applications
  to exponential time problems.
\newblock In {\em 42nd International Colloquium on Automata, Languages and
  Programming (ICALP 2015)}, volume 9134, pages 231--242. Springer, 2015.
\newblock \href {http://dx.doi.org/10.1007/978-3-662-47672-7_19}
  {\path{doi:10.1007/978-3-662-47672-7_19}}.

\bibitem{Bjorklund09}
Andreas Bj\"{o}rklund, Thore Husfeldt, and Mikko Koivisto.
\newblock Set partitioning via inclusion-exclusion.
\newblock {\em SIAM J. Comput.}, 39(2):546--563, July 2009.
\newblock \href {http://dx.doi.org/10.1137/070683933}
  {\path{doi:10.1137/070683933}}.

\bibitem{bjor17}
Andreas Bj\"{o}rklund, Thore Husfeldt, Kaski Ptteri, and Mikko Koivisto.
\newblock Narrow sieves for parameterized paths and packings.
\newblock {\em Journal of Computer and System Sciences}, 87:119 -- 139, 2017.
\newblock \href {http://dx.doi.org/10.1016/j.jcss.2017.03.003}
  {\path{doi:10.1016/j.jcss.2017.03.003}}.

\bibitem{bjor16}
Andreas Bj\"{o}rklund, Petteri Kaski, and {\L}ukasz Kowalik.
\newblock Constrained multilinear detection and generalized graph motifs.
\newblock {\em Algorithmica}, 74(2):947--967, 2016.
\newblock \href {http://dx.doi.org/10.1007/s00453-015-9981-1}
  {\path{doi:10.1007/s00453-015-9981-1}}.

\bibitem{cygan16}
Marek Cygan, Holger Dell, Daniel Lokshtanov, D\'{a}niel Marx, Jesper Nederlof,
  Yoshio Okamoto, Ramamohan Paturi, Saket Saurabh, and Magnus Wahlstr\"{o}m.
\newblock On problems as hard as {CNF-SAT}.
\newblock {\em ACM Transactions on Algorithms}, 12(3):41:1--41:24, 2016.
\newblock \href {http://dx.doi.org/10.1145/2925416}
  {\path{doi:10.1145/2925416}}.

\bibitem{CFGKMP16}
Marek Cygan, Fedor~V. Fomin, Alexander Golovnev, Alexander~S. Kulikov, Ivan
  Mihajlin, Jakub Pachocki, and Arkadiusz Soca{\l}a.
\newblock Tight bounds for graph homomorphism and subgraph isomorphism.
\newblock In {\em 27th Annual ACM-SIAM Symposium on Discrete Algorithms}, SODA
  '16, pages 1643--1649. SIAM, 2016.
\newblock \href {http://dx.doi.org/10.1137/1.9781611974331.ch112}
  {\path{doi:10.1137/1.9781611974331.ch112}}.

\bibitem{Dessmark00}
Anders Dessmark, Andrzej Lingas, and Andrzej Proskurowski.
\newblock Faster algorithms for subgraph isomorphism of $k$-connected partial
  $k$-trees.
\newblock {\em Algorithmica}, 27(3):337--347, January 2000.
\newblock \href {http://dx.doi.org/10.1007/s004530010023}
  {\path{doi:10.1007/s004530010023}}.

\bibitem{fomin04}
Fedor~V. Fomin, Dieter Kratsch, and Gerhard~J. Woeginger.
\newblock Exact (exponential) algorithms for the dominating set problem.
\newblock In {\em 30th International Conference on Graph-Theoretic Concepts in
  Computer Science}, WG'04, pages 245--256. Springer-Verlag, 2004.
\newblock \href {http://dx.doi.org/10.1007/978-3-540-30559-0_21}
  {\path{doi:10.1007/978-3-540-30559-0_21}}.

\bibitem{hardy1918}
Godfrey~H. Hardy and Srinivasa Ramanujan.
\newblock Asymptotic formula{\ae} in combinatory analysis.
\newblock {\em Proceedings of the London Mathematical Society},
  s2-17(1):75--115, 1918.
\newblock \href {http://dx.doi.org/10.1112/plms/s2-17.1.75}
  {\path{doi:10.1112/plms/s2-17.1.75}}.

\bibitem{held1961}
Michael Held and Richard~M. Karp.
\newblock A dynamic programming approach to sequencing problems.
\newblock In {\em Proceedings of 16th ACM National Meeting}, ACM '61, pages
  71.201--71.204. ACM, 1961.
\newblock \href {http://dx.doi.org/10.1145/800029.808532}
  {\path{doi:10.1145/800029.808532}}.

\bibitem{ImpaSETH}
Russell Impagliazzo and Ramamohan Paturi.
\newblock On the complexity of k-{SAT}.
\newblock {\em Journal of Computer and System Sciences}, 62(2):367--375, March
  2001.
\newblock \href {http://dx.doi.org/10.1006/jcss.2000.1727}
  {\path{doi:10.1006/jcss.2000.1727}}.

\bibitem{Impa01spar}
Russell Impagliazzo, Ramamohan Paturi, and Francis Zane.
\newblock Which problems have strongly exponential complexity?
\newblock {\em Journal of Computer and System Sciences}, 63(4):512--530, 2001.
\newblock \href {http://dx.doi.org/10.1006/jcss.2001.1774}
  {\path{doi:10.1006/jcss.2001.1774}}.

\bibitem{Iwata2015}
Yoichi Iwata and Yuichi Yoshida.
\newblock On the equivalence among problems of bounded width.
\newblock In {\em 23rd Annual European Symposium on Algorithms (ESA 2015)},
  pages 754--765. Springer, 2015.
\newblock \href {http://dx.doi.org/10.1007/978-3-662-48350-3_63}
  {\path{doi:10.1007/978-3-662-48350-3_63}}.

\bibitem{karp1972}
Richard~M. Karp.
\newblock {\em Reducibility among Combinatorial Problems}, pages 85--103.
\newblock The IBM Research Symposia Series. Springer US, 1972.
\newblock \href {http://dx.doi.org/10.1007/978-1-4684-2001-2_9}
  {\path{doi:10.1007/978-1-4684-2001-2_9}}.

\bibitem{Koivisto09}
Mikko Koivisto.
\newblock Partitioning into sets of bounded cardinality.
\newblock In {\em Parameterized and Exact Computation ({IWPEC} 2009)}, volume
  5917 of {\em Lecture Notes in Computer Science}, pages 258--263.
  Springer-Verlag, 2009.
\newblock \href {http://dx.doi.org/10.1007/978-3-642-11269-0_21}
  {\path{doi:10.1007/978-3-642-11269-0_21}}.

\bibitem{Koutis16}
Ioannis Koutis and Ryan Williams.
\newblock {LIMITS} and applications of group algebras for parameterized
  problems.
\newblock {\em ACM Trans. Algorithms}, 12(3):31:1--31:18, May 2016.
\newblock \href {http://dx.doi.org/10.1145/2885499}
  {\path{doi:10.1145/2885499}}.

\bibitem{kow16}
{\L}ukasz Kowalik and Juho Lauri.
\newblock On finding rainbow and colorful paths.
\newblock {\em Theoretical Computer Science}, 628(C):110--114, 2016.
\newblock \href {http://dx.doi.org/10.1016/j.tcs.2016.03.017}
  {\path{doi:10.1016/j.tcs.2016.03.017}}.

\bibitem{kri17}
R.~Krithika, Abhishek Sahu, and Prafullkumar Tale.
\newblock Dynamic parameterized problems.
\newblock In {\em 11th International Symposium on Parameterized and Exact
  Computation (IPEC 2016)}, volume~63 of {\em Leibniz International Proceedings
  in Informatics (LIPIcs)}, pages 19:1--19:14. Schloss
  Dagstuhl--Leibniz-Zentrum fuer Informatik, 2017.
\newblock \href {http://dx.doi.org/10.4230/LIPIcs.IPEC.2016.19}
  {\path{doi:10.4230/LIPIcs.IPEC.2016.19}}.

\bibitem{Lingas89}
Andrzej Lingas.
\newblock Subgraph isomorphism for biconnected outerplanar graphs in cubic
  time.
\newblock {\em Theoretical Computer Science}, 63(3):295--302, 1989.
\newblock \href {http://dx.doi.org/10.1016/0304-3975(89)90011-X}
  {\path{doi:10.1016/0304-3975(89)90011-X}}.

\bibitem{lovasz2009}
L{\'a}szl{\'o} Lov{\'a}sz and Michael~D Plummer.
\newblock {\em Matching theory}, volume 367.
\newblock American Mathematical Society, 2009.

\bibitem{marx14}
D{\'a}niel Marx and Michal Pilipczuk.
\newblock {Everything you always wanted to know about the parameterized
  complexity of Subgraph Isomorphism (but were afraid to ask)}.
\newblock In {\em 31st International Symposium on Theoretical Aspects of
  Computer Science (STACS 2014)}, volume~25 of {\em Leibniz International
  Proceedings in Informatics (LIPIcs)}, pages 542--553. Schloss
  Dagstuhl--Leibniz-Zentrum fuer Informatik, 2014.
\newblock \href {http://dx.doi.org/10.4230/LIPIcs.STACS.2014.542}
  {\path{doi:10.4230/LIPIcs.STACS.2014.542}}.

\bibitem{Matousek92}
Ji{\v{r}}{\'\i} Matou{\v{s}}ek and Robin Thomas.
\newblock On the complexity of finding iso- and other morphisms for partial
  $k$-trees.
\newblock {\em Discrete Mathematics}, 108(1):343 -- 364, 1992.
\newblock \href {http://dx.doi.org/10.1016/0012-365X(92)90687-B}
  {\path{doi:10.1016/0012-365X(92)90687-B}}.

\bibitem{neder16}
Jesper Nederlof.
\newblock Finding large set covers faster via the representation method.
\newblock In {\em 24th Annual European Symposium on Algorithms (ESA 2016)},
  volume~57 of {\em Leibniz International Proceedings in Informatics (LIPIcs)},
  pages 69:1--69:15. Schloss Dagstuhl--Leibniz-Zentrum fuer Informatik, 2016.
\newblock \href {http://dx.doi.org/10.4230/LIPIcs.ESA.2016.69}
  {\path{doi:10.4230/LIPIcs.ESA.2016.69}}.

\bibitem{nijen78}
Albert Nijenhuis and Herbert~S. Will.
\newblock {\em Combinatorial Algorithms: For Computers and Hard Calculators}.
\newblock Academic Press, 2nd edition, 1978.

\bibitem{Trabelsi18}
Ohad Trabelsi.
\newblock Nearly optimal time bounds for k{P}ath in hypergraphs.
\newblock {\em CoRR}, 2018.
\newblock URL: \url{http://arxiv.org/abs/1803.04940}.

\bibitem{woeg03}
Gerhard~J. Woeginger.
\newblock Exact algorithms for {NP}-hard problems: A survey.
\newblock In Michael J\"{u}nger, Gerhard Reinelt, and Giovanni Rinaldi,
  editors, {\em Combinatorial Optimization - Eureka, You Shrink!}, pages
  185--207. Springer-Verlag, 2003.
\newblock \href {http://dx.doi.org/10.1007/3-540-36478-1}
  {\path{doi:10.1007/3-540-36478-1}}.

\end{thebibliography}


\end{document}